\documentclass[letter,11pt]{article}
\usepackage{mathpazo}
\usepackage{multirow}
\usepackage{nicefrac}
\usepackage{setspace}
\usepackage{thm-restate}
\usepackage[margin=1in]{geometry}
\usepackage{xcolor}
\definecolor{xlinkcolor}{cmyk}{1,1,0,0}
\usepackage{hyperref}
\hypersetup{colorlinks=true,linkcolor=black,citecolor=xlinkcolor}
\usepackage{color}
\usepackage{amsthm}
\usepackage{graphicx}
\usepackage{amsmath}
\usepackage{amssymb}
\usepackage{bbm}
\usepackage{dsfont}

\usepackage{diagbox}
\usepackage{makecell}
\usepackage{float}

\usepackage[]{algorithm2e}

\newtheorem{theorem}{Theorem}[section]

\newtheorem{corollary}[theorem]{Corollary}
\newtheorem{lemma}[theorem]{Lemma}

\newtheorem{claim}[theorem]{Claim}

\newtheorem{definition}[theorem]{Definition}

\newtheorem{remark}[theorem]{Remark}

\newcommand{\poly}{\text{poly}}
\newcommand{\logstar}{\log^{*}}

\newcommand{\eps}{\varepsilon}

\renewcommand{\tilde}{\widetilde}

\def\ed{\mathsf{ed}}

\usepackage{enumerate}

\newcommand{\HSC}{$\mathsf{HSC}$\xspace}

\newcommand{\EFEEOV}{$\exists \forall \exists \exists\mathsf{OVH}$\xspace}

\title{On Complexity of 1-Center in Various Metrics	}  

\author{Amir Abboud\footnote{This project has received funding from the European Research Council (ERC) under the European Union’s Horizon Europe research and innovation programme (grant agreement No 101078482). Additionally, the author is supported by an Alon scholarship and a research grant from the Center for New Scientists at the Weizmann Institute of Science.}\\  Weizmann Institute of Science\\ \texttt{amir.abboud@weizmann.ac.il}
\and MohammadHossein Bateni\\ Google Research\\ \texttt{bateni@google.com}\and
Vincent Cohen-Addad\\ Google Research\\ \texttt{vcohenad@gmail.com}\and
Karthik C.\ S.\footnote{This work was supported by Subhash Khot's Simons Investigator Award and by a grant from the Simons Foundation, Grant Number 825876, Awardee Thu D. Nguyen  and  by the National Science Foundation under Grant CCF-2313372.}\\ Rutgers University \\ \texttt{karthik.cs@rutgers.edu}\and 
Saeed Seddighin\footnote{Supported by a Google research gift.}\\ Toyota Technological Institute at Chicago\\  \texttt{saeedreza.seddighin@gmail.com} }

\date{}

\newcommand{\ground}{\Sigma^*}
\newcommand{\dist}{D}
\newcommand{\nph}{NP-hard}
\newcommand{\woh}{$W[1]$-hard}
\newcommand{\ptas}{PTAS}
\newcommand{\reals}{\mathbb{R}}

\begin{document}

\maketitle

\begin{abstract}
  We consider the classic 1-center problem: Given a set $P$ of $n$ points in a metric space
  find the point in $P$ that minimizes the maximum distance to the other points of $P$.
  We study the complexity of this problem in $d$-dimensional $\ell_p$-metrics and in edit and Ulam metrics over strings of length $d$.
  Our results for the 1-center problem may be classified based on $d$ as follows.

\begin{itemize}
\item \textbf{Small $d$.} 
Assuming the hitting set conjecture (\HSC), we  show that when $d=\omega(\log n)$, no subquadratic algorithm can solve the 1-center problem in any of the $\ell_p$-metrics, or in the edit or Ulam metrics.

\item \textbf{Large $d$.} 
When $d=\Omega(n)$, we extend our conditional lower bound to rule out sub\emph{quartic} algorithms for the 1-center problem in edit metric (assuming Quantified SETH).
On the other hand, we give a $(1+\epsilon)$-approximation for 1-center in the Ulam metric
with running time $\tilde{O_{\eps}}(nd+n^2\sqrt{d})$.  
\end{itemize}

We also strengthen some of the above lower bounds by allowing approximation algorithms or by reducing the dimension $d$, but only against a weaker class of algorithms which list all requisite solutions.
Moreover, we  extend one of our hardness results to rule out subquartic algorithms for  the well-studied 1-median problem in the edit metric, where given a set of $n$ strings each of length $n$, the goal is to find a string in the set that minimizes the sum of the edit distances to the rest of the strings in the set.
\end{abstract}

\clearpage

\section{Introduction}
Given a set of points $P$ in a metric space, finding the point that ``best'' represents
$P$ is a fundamental question in both discrete and continuous optimization. Motivated by
applications ranging from machine learning to computational biology, this question has
naturally received a large amount of attention through the years.

The objective can be phrased in various ways: In the median problem, the goal is to find
the point $p$ that minimizes the sum of the distances to the points in $P$; in the mean
problem, it is the point that minimizes the sum of distances squared; while in the center
problem, it is the point $p$ that minimizes the maximum distance from a point of $P$ to $p$.
When the metric is the $\ell_2$ (Euclidean) metric, the question of computing the geometric median 
dates back to the 17th century, when Torricelli was looking for a solution for the case $|P| = 3$,
and to whom Fermat described an explicit solution.
More recently, the question of computing the center has also become central in applications arising,
e.g., in machine learning~\cite{ClarksonHW10}, to compute the minimum enclosing ball of a set
of points, or in computational biology, to find a good representative of a set of strings (representing
molecular sequences) (e.g.,~\cite{nicolas2005hardness}). This fundamental computational
geometry problem which has applications to various domains,  is the problem we consider in this paper.

Formally, in the (often referred to as the \emph{discrete}) 1-center problem, the input is a set of
points $P$ in a metric space, and the goal is to find a point of $P$ that minimizes
the maximum distance to the points in $P$. When doing data summarization or compression, the discrete
version often makes more sense: Given a set of, say $n$ strings, taking the most representative
string among the input strings, or at least in the set of grammatically (or semantically) meaningful strings
is much more insightful than taking an arbitrary string as representative. This also applies more globally,
outputting a data element that has been observed provides a better summary than a data element that has been
forged by the algorithm and that may be unlikely to exist in the real-world.
From a computational complexity standpoint, this problem can be easily solved in time $O(|P|^2 f(d))$ where $f(d)$ is the
time required to compute the distance of two points.
This can be done by enumerating all possible choices for the center; and for each choice computing the distance from each
point in $P$; then outputting the best center. However, is this na\"\i{}ve algorithm the best we can do?

The computational geometry community has done extensive work on the above question since the 80s. For metrics such as $\ell_1$ or
$\ell_2$, computing the center has received a large deal of attention. When the dimension is
assumed to be a constant, there exist \emph{barely subquadratic} algorithms for the $\ell_2$ metric,
while there exists near-linear time algorithms for the $\ell_1$ case (for a discussion
on this we refer the reader to~\cite{W18}).
For the case of string
metrics, such as Ulam or Edit distance metrics, nothing better than the $O(|P|^2 f(d))$
(where $d$ is the string length) ``brute-force'' algorithm is known.

Understanding how fast the 1-center problem can be solved in these different metrics is not only
interesting from a computational complexity point of view, but also from the perspective of an improved understanding of
 the geometry of these metrics. For example, is the geometry of the $\ell_1$ metric ``more amenable''
for designing algorithms than the $\ell_2$ one? Is the Edit distance metric hard for such problems?
We also believe that understanding the geometry of the  Ulam
and Edit distance metrics, which one may interpret as generalization of the Hamming metric,   is not only a very basic computational geometry question, but also would  likely lead to better algorithms for these widely-studied problems. 
We thus ask:

\begin{center}\emph{How fast can the 1-center problem be solved or approximated in \\ $\ell_p$-metrics and stringology metrics such as Ulam or Edit distance?}\end{center}

\subsection{Our Results}
In this paper, we take a step towards answering the above question by providing lower and upper bounds
on solving the 1-center problem  in $\ell_p$, Ulam, and Edit distance
metrics.

Assuming the Hitting Set Conjecture (\HSC), we provide a strong conditional lower bound for the 1-center problem, in a pleathora of metrics.

\begin{restatable}{theorem}{ullower}(see Theorem~\ref{thm:ellp}, Corollary~\ref{cor:ulamsubquad}, and Corollary~\ref{cor:editsubquad} for formal statement)
  \label{thm:ulam:center:LB}
  Assuming \HSC, no algorithm running in time $n^{2-o(1)}$ can given as input a set of points/strings $P$ of dimension/length $d$ solve the discrete 1-center problem in Edit/Ulam/$\ell_p$ metric, where $|P|=n$, $d=\tilde\Omega(\log n)$, and $p\in \mathbb{R}_{\ge 1}\cup\{0\}$.
\end{restatable}

Moreover,  by assuming a stronger complexity theoretic assumption we can strengthen this lower bound in the case when $d=\poly(n)$ for the Edit metric. For the sake of presentation, we state our result below when $d=n$. 

\begin{restatable}{theorem}{stronglower}(see Theorem~\ref{thm:Edit})
  \label{thm:ulameditLB}
  Assuming Quantified SETH, no algorithm running in time $n^{4-o(1)}$ can given as input a set of $n$ strings of length $d:=n$ each, solve the discrete 1-center problem in Edit  metric.
\end{restatable}

 It's worth emphasizing that the above lower bound for the edit metric  is a rare quartic lower bound in fine-grained complexity. It’s true that, conceptually, it’s not unexpected because there’s a quadratic hardness from the 1-center problem and a quadratic hardness from edit distance, so we would expect the combined problem to be quartic. But we find it noteworthy that this actually works on the technical level because complexity theory is full of notorious examples where such “semi-stitching techniques” completely fail (for example KRW games \cite{KRW91}).

Note that we cannot expect such lower bounds for the 1-center problem in $\ell_p$-metrics when $d=n$, as one can compute all pairwise distances within a point-set in subcubic time using fast matrix multiplication.
 
Next, we complement the lower bounds by the following subcubic approximation scheme for the 1-center in the Ulam metric. 
 
\begin{restatable}{theorem}{upperl}
  \label{thm:ulam:upper}
  There exists a $1+\epsilon$ approximation algorithm for the 1-center under Ulam metric that runs in time $\tilde O_{\epsilon}(nd + n^2\sqrt{d})$.
\end{restatable}
It is worth emphasizing here that for the (discrete) 1-center problem in any metric space, an arbitrary point in the input is a 2-approximate solution. Also note that exact $1$-center in Ulam metric can be solved in $O(n^2d)$ time. It remains an open problem to show a conditional lower bound of $n^{3-o(1)}$ for computing the 1-center in the Ulam metric for $n$ strings each of length $n$. 

Finally, we strengthen some of the lower bounds above, but against a weaker class of algorithms, specifically, against algorithms which list all requisite solutions.
Using the ideas in \cite{W18,C20}, assuming \HSC, we rule out subquadratic algorithms that can list all optimal solutions to the 1-center problem in the Euclidean metric even for very low $d=o(\log{n})$ dimensions. At a high level, this result contrasts with both $\ell_1$ and $\ell_\infty$ metrics where the 1-center in $o(\log{n})$ dimensions can be solved in $n^{1+o(1)}$ time. 

\begin{theorem}[see Theorem~\ref{thm:euclidean} for formal statement]
\label{thm:l2center}
Assuming \HSC,
there is no $n^{2-o(1)}$-time algorithm listing all optimal solutions to the 1-center problem in $7^{\logstar n}$ dimensions in the Euclidean metric. 
\end{theorem}

In the same spirit as above, by applying the distributed PCP framework \cite{ARW17,R18} we extend the lower bound in  Theorem~\ref{thm:ulam:center:LB} against approximation algorithms which list all approximately optimal 1-centers. 
\begin{theorem}[see Theorem~\ref{thm:hardnessapprox} for formal statement]
\label{thm:lpinapprox}
Assuming \HSC, there is some $\delta>0$, such that no $n^{2-o(1)}$-time algorithm can given as input a set of points/strings $P$ of dimension/length $d$, list all $(1+\delta)$-approximate solutions to the 1-center problem in the Edit/Ulam/$\ell_p$ metric, where $|P|=n$, $d=\tilde\Omega(\log n)$, and $p\in \mathbb{R}_{\ge 1}\cup\{0\}$.
\end{theorem}
One may compute all pairwise distances in $\tilde{O}(n^2)$ time for the inputs in Theorems~\ref{thm:l2center}~and~\ref{thm:lpinapprox}, and then obtain the list of all optimal  and approximately optimal solutions efficiently. Our theorems above say that one cannot do much better. It remains an intriguing open problem to extend the above two conditional lower bounds but against standard decision algorithms. We note that this involves breaking some technical barriers and in particular, developing techniques that go beyond the dimensionality reduction techniques of \cite{W18,C20} and the distributed PCP framework \cite{ARW17,R18} respectively.

We close this subsection by a short discussion about the Discrete 1-median problem in $\ell_p$-metrics and string metrics. For the case when $d=n$, we can prove a result similar to Theorem~\ref{thm:ulameditLB} (see Remark~\ref{rem:1med}). On the other hand for $\ell_p$-metrics, one cannot prove a result similar to Theorem~\ref{thm:ulam:center:LB} for the 1-median problem, because the 1-median problem in Hamming and $\ell_1$-metrics admits a near linear time algorithm and for the Euclidean metric, it is even unclear if the problem is in NP! (see discussion in \cite{GGJ76}.) Also note that by subsampling coordinates, we can approximate 1-median in all $\ell_p$-metrics to $(1+\varepsilon)$ factor, for any $\varepsilon>0$ in near linear time. 

\subsection{Related Work}
We now review the related work on the 1-center problem, and the related
1-median problem.
Both problems may be considered in the \emph{discrete} or \emph{continuous}
settings.  The discrete\footnote{Sometimes called the \textit{medoid}
problem in contrast to \textit{median}, \emph{generalized} median,
or \emph{Steiner} string.} version asks the center or median to be picked
from an input set of points, while in the continuous version, the ``center'' is an
arbitrary element of the metric. See \cite{DBLP:conf/soda/Cohen-AddadSL21} for a discussion on these two settings.

Below, we mainly discuss 1-center problem in stringology metrics as the literature on related work in $\ell_p$ metrics is too vast to survey (but the interested reader may look at \cite{CLMPS16,LiMW02,LLMWZ03,DBLP:journals/mst/FrancesL97} and the references therein).

\paragraph*{Metrics arising in stringology}
We now review results on the 1-center problem in metric spaces arising
from stringology applications.
Let $\Sigma$ be an alphabet, often the binary alphabet $\{0, 1\}$.
Consider the set of strings $\ground = \Sigma^L$ of length $L$, with
a metric distance $\dist : \ground \times \ground \mapsto \reals$. Researchers
have mainly considered the following metrics defined over this space:
\begin{itemize}
  \item \textbf{Edit distance} (ED or Levenshtein distance): The minimum number
  of single-character insertions, deletions, and substitutions
  required to change one string to the other.
  
\item \textbf{Hamming distance or $\ell_1$ over binary alphabets} (HD):
  A special case of edit distance, where only
  substitutions are allowed.
  
  \item \textbf{Ulam distance} (UD): Same as edit distance with the restriction
  that the input strings may not contain any character more than once.

\end{itemize}

For most of the above metrics, one need to incorporate into the running times obtained for
simpler metric such as $\ell_1$ or $\ell_2$ the time it takes to compute the exact or
approximate distance between any two points of the space.
Naumovitz et al.~\cite{naumovitz2017accurate} show how to approximate UD within factor $1+\eps$ in
time $\tilde O(d/\eta + \sqrt d)$ if the distance is $\eta$.  This result is
tight up to log factors.


Turning back to the 1-center and 1-median problem, note the \emph{discrete} versions
can be solved exactly via $O(n^2)$ distance computations, giving trivial $\tilde O(n^2 d)$-time
algorithms for the case of UD.
In Section~\ref{sec:ulam-up}, we show that
this can be improved to
$\tilde O(n^2\sqrt d)$ for $1+\eps$ approximation if we combine two
algorithms~\cite{MS20,naumovitz2017accurate} for computations of UD.  Note that a
2-approximation is trivial, as we can output any string as the hub in the case of 1-center
or a random string in the case of 1-median.

Recently \cite{CDK21} made progress on obtaining better approximation algorithms for
the continuous 1-Median problem in UD, where the median can be picked from anywhere in space,
by presenting the first polynomial-time constant-factor  approximation
algorithm with approximation guarantee smaller than 2 as well as an exact algorithm for
the case where the input contains \emph{three} strings.
They observe that if the average distance to median is $\Omega(d)$,
picking the best string as the median already gives an approximation
better than two.  Now the problem is reduced to the above case if the
total cost is mostly due to a small subset of characters.  Otherwise,
they argue that one can deduce the relative ordering of a good portion
of the optimal median by looking at pairs of characters whose relative
order is consistent in \emph{most} input strings.

Note that in the continuous case, for constant $d$ or constant $n$, the median and center
problems are both solvable in polynomial time for string problems.
De la Higuera and Casacuberta~\cite{HC00} prove that median and center are both \nph.
%
Nicolas and Rivals~\cite{NR03} lift the restrictions and show that median and center
are both \nph\ and \woh\ (when parameterized by $n$, the number of
strings), even for binary alphabets.
%
Prior to these works, \nph ness of median was only known for ED when
the substitutions have specific costs for each pair of
characters~\cite{SP03}.
Li et al.~\cite{LMW02a,LiMW02} give a \ptas\ for the HD 1-center problem,%
\footnote{This is the \textit{closest string} problem.  They also give
a \ptas\ for the \textit{closest substring} problem, which assumes that the cost
of deletions from the input strings ($\infty$ for HD) is zero.}
augmenting the LP-based \ptas\ for super-logarithmic
$d$~\cite{BLPR97}.  The HD 1-center problem  is known to
be \nph~\cite{DBLP:journals/mst/FrancesL97,LLMWZ03}.%
\footnote{Note that median is solvable exactly for HD.}
Previously the best polynomial-time approximation ratio was $\frac43+\eps$
in general~\cite{LLMWZ03,GJL99}, with an exact algorithm known for
constant $d$ (optimal value)~\cite{SBGHM97}. 

\subsection{Organization of the Paper}
In Section~\ref{sec:prelim}, we provide the formal definition of 1-center problem and the various hypotheses used in the paper. 
In
Section~\ref{sec:exactLB}, we provide conditional lower bounds against exact algorithms that compute 1-center when $d=\omega(\log n)$.
Next,
in Section~\ref{sec:ulam-up}, we provide a subcubic approximation algorithm for 1-center in Ulam metric when $d=n$.
Finally, in Section~\ref{sec:inapprox}, we provide some hardness of approximation results (Theorem~\ref{thm:lpinapprox}).

\section{Preliminaries}\label{sec:prelim}

\begin{definition}[Discrete 1-center]
Let $(X,\Delta)$ be a metric space. Given a set of points $P\subseteq X$ in  the metric space, find $x$ in $P$ which minimizes the maximum distance to every other point. 
\end{definition}

Perhaps the most popular assumption for proving conditional lower bounds for polynomial time problems is the Orthogonal Vectors Hypothesis (OVH) that is implied by the Strong Exponential Time Hypothesis (SETH). 
Unfortunately, the logical structure of these problems makes reductions to our 1-center problems difficult.
This was observed already by Abboud, Vassilevska Williams, and Wang \cite{AWW16} in the context of 1-center in \emph{graphs} (known as the Graph Radius problem) and has lead them to introduce the hitting set conjecture (\HSC): a stronger variant of OVH that facilitates reductions to problems with different structure. 
A formal barrier for establishing \HSC (and similarly also any hardness results for 1-center) under SETH was presented by Carmosino et al. \cite{CGIMPS16}.

\begin{definition}[\HSC]
For every $\eps>0$ there exists $c>1$ such that no algorithm running in time $n^{2-\eps}$ can, given as input two collections of $n$-many subsets $\mathcal{A}$ and $\mathcal{B}$ of the universe $U:=[c \log n]$, determine if there exists $S$ in $\mathcal{A}$ which has non-empty intersection with every subset in $\mathcal{B}$. 
\end{definition}

The difference between \HSC and OVH is in the quantifiers: $\exists \forall$ versus $\exists \exists$. 
Studying the \emph{polyline simplification} problem, Bringmann and Chaudhury \cite{BC19} proposed a further strengthening with more quantifiers.
Just like OVH is implied by SETH, an assumption about $k$-SAT, so too can \HSC and its generalizations with more quantifiers be based on the hardness of a quantified version of $k$-SAT; an assumption called \emph{Quantified-SETH}.
Interestingly, the previous papers using Quantified-SETH \cite{BC19,ABHS20} only needed its special case where the quantifier structure is $ \forall \exists$; whereas in this paper we benefit from a $\exists \forall \exists$ structure that has one more alternation. 

The specific hardness assumption (implied by Quantified SETH) that we need is the following; we refer to \cite{BC19,ABHS20} for further discussion on Quantified-SETH and to \cite{AWW16,williams2018some} for further discussion on \HSC and on the need for assumptions with other quantifier structures.

%

\begin{definition}[\EFEEOV]
For every $\eps>0$ there exists $c>1$ such that no algorithm running in time $n^{4-\eps}$ can, given as input four collections of $n$-many subsets $\mathcal{A},\mathcal{B},\mathcal{C},$ and $\mathcal{D}$ of the universe $U:=[c \log n]$, determine if there exists $S_A$ in $\mathcal{A}$ such that for all $S_B$ in $\mathcal{B}$ there exist $S_C \in \mathcal{C}$ and $S_D \in \mathcal{D}$ such that the intersection $S_A \cap S_B \cap S_C \cap S_D=\emptyset$ is empty. 
\end{definition}

\section{Exact Lower Bounds for 1-center}
\label{sec:exactLB} 

In this section, we prove conditional lower bounds for the 1-center problem.
We start with some high-level remarks about the reductions and our contributions.

Previous work (for example \cite{R18,DKL19}) has already designed reductions from SETH and OVH to \emph{closest pair} kinds of questions for the metrics we consider, and our work can be viewed as lifting these results to the 1-center question.
As discussed in Section~\ref{sec:prelim} this requires a new starting assumption (either Quantified SETH or the Hitting Set Conjecture) that has a different structure.
Thus, technically, the main contribution is to adapt the gadgetry of previous work into new reductions with a different structure. In some cases, fundamental difficulties arise and we can only resolve them by requiring that the algorithm lists all solutions. 

 In all our reductions, we first reduce to the \emph{Discrete 1-center with Facilities}, where given  a set of clients $C\subseteq X$ and a set of facilities $F\subseteq X$ in  the metric space, the goal is to find $x$ in $F$ which minimizes the maximum distance to every point in $C$.  We then reduce a hard instance $(F,C)$ of the Discrete 1-center with Facilities problem to an instance $P$ of the standard Discrete 1-center problem (without facilities) by adding a few additional coordinates to points in $F\cup C$ and then introducing a new point/string $s$ such that it is far from every point in $C$ (in comparison to its distance from the points/strings in $F$). Thus we ensured that the 1-center of $P:=F\cup C\cup \{s\}$ must be from $F$. Nevertheless, for the sake of compactness, this two step reduction in the proofs of this section is sometimes written as a one step reduction.

This section is organized as follows.
In Section~\ref{sec:hscellp}, we show the conditional subquadratic time lower bounds for 1-center in various metrics (Theorem~\ref{thm:ulam:center:LB}). Next, in Section~\ref{sec:edit}, we show the conditional subquartic time lower bound for 1-center in edit metric (also Theorem~\ref{thm:ulameditLB}) and explain how to adapt it for 1-median.
Finally, in Section~\ref{sec:euclid}, we show that there are no subquadratic listing algorithms for Euclidean 1-center even in low dimensions (Theorem~\ref{thm:l2center}).

\subsection{Subquadratic Lower Bounds for 1-center when $d=\omega(\log n)$ in String and $\ell_p$-metrics}
\label{sec:hscellp}

In this subsection, we show that  subquadratic time algorithms for 1-center do not exist in $\ell_p$-metrics, Ulam metic, and edit metric, when $d=\omega(\log n)$.

\begin{theorem}[Subquadratic Hardness of 1-center in $\ell_p$-metrics]\label{thm:ellp}
Let $p\in \mathbb{R}_{\ge 1}\cup\{0\}$.
Assuming \HSC, for every $\eps>0$, there exists $c>1$ such that no algorithm running in time $n^{2-\eps}$ can, given as input a point-set $P\subseteq \{0,1\}^d$, solve the discrete 1-center problem in $\ell_p$-metric, where $|P|=n$ and $d=c\log n$.
\end{theorem}
\begin{proof}
Let $(\mathcal{A}:=(S_1,\ldots ,S_n),\mathcal{B}:=(T_1,\ldots ,T_n),U)$ be an instance arising from \HSC. We construct a point-set $P\subseteq \{0,1\}^d$ where $|P|=2n+1$ and  $d=5\cdot |U|+2$. We build the two maps $\tau_{\mathcal{A}}:\mathcal{A}\to\{0,1\}^d,\tau_{\mathcal{B}}:\mathcal{B}\to\{0,1\}^d$ and a special point $s\in \{0,1\}^{d}$ and the point-set $P$ is then simply defined to be the union of $\{s\}$ and the images (range) of $\tau_{\mathcal{A}}$ and $\tau_{\mathcal{B}}$, i.e., 
\[P:=\{\tau_{\mathcal{A}}(S)\mid S\in\mathcal{A}\}\cup \{\tau_{\mathcal{B}}(T)\mid T\in\mathcal{B}\}\cup\{s\}.\]

Let $U:=\{u_1,\ldots ,u_m\}$. We define our special point $s$ as follows:
\[\forall i\in[5m+2],\ s_i:=
\begin{cases}
0&\text{ if }1\le i\le 3m\\
1&\text{ if }3m+1\le i\le 5m+2
\end{cases}
\]

For any $S\in \mathcal{A}$ we define $\tau_{\mathcal{A}}(S)$ as follows:
\[\forall i\in[5m+2],\ \tau_{\mathcal{A}}(S)_i:=
\begin{cases}
1&\text{ if }u_i\in S \text{ and }1\le i\le m\\
0&\text{ if }u_i\notin S \text{ and }1\le i\le m\\
0&\text{ if }u_{i-m}\in S \text{ and }m+1\le i\le 2m\\
1&\text{ if }u_{i-m}\notin S \text{ and }m+1\le i\le 2m\\
0&\text{ if }2m+1\le i\le 4m+1\\
1&\text{ if }4m+2\le i\le 5m+2
\end{cases}
\]

For any $T\in \mathcal{B}$ we define $\tau_{\mathcal{B}}(T)$ as follows:\allowdisplaybreaks
\[\forall i\in[5m+2],\ \tau_{\mathcal{B}}(T)_i:=
\begin{cases}
1&\text{ if }u_i\in T \text{ and }1\le i\le m\\
0&\text{ if }u_i\notin T \text{ and }1\le i\le m\\
0&\text{ if }m+1\le i\le 2m\\
0&\text{ if }u_{i-2m}\in T \text{ and }2m+1\le i\le 3m\\
1&\text{ if }u_{i-2m}\notin T \text{ and }2m+1\le i\le 3m\\
0&\text{ if }3m+1\le i\le 5m+2
\end{cases}
\]

Notice that for any $S,S'$ in $\mathcal{A}$ and $T$ in $\mathcal{B}$, we have \[\|\tau_{\mathcal{A}}(S)-\tau_{\mathcal{B}}(T)\|_p=\left(|S|+|T|-2\cdot |S\cap T|+m-|S|+m-|T|+m+1\right)^{1/p}=\left(3m+1-2\cdot |S\cap T|\right)^{1/p}.\]
\[
\|\tau_{\mathcal{A}}(S)-\tau_{\mathcal{A}}(S')\|_p\le \left(2m\right)^{1/p}.
$$
$$\|\tau_{\mathcal{A}}(S)-s\|_p= \left(2m+1\right)^{1/p}.$$
$$\|\tau_{\mathcal{B}}(T)-s\|_p= \left(3m+2\right)^{1/p}.\]

Suppose there exists $S$ in $\mathcal{A}$ such that it intersects with every subset $T$ in $\mathcal{B}$ then $\tau_{\mathcal{A}}(S)$ has  distance strictly less than $(3m+1)^{1/p}$ with $\tau_{\mathcal{B}}(T)$ for every $T$ in $\mathcal{B}$. Additionally,  $\tau_{\mathcal{A}}(S)$ has distance at most $(2m)^{1/p}$ with $\tau_{\mathcal{A}}(S')$ for any $S' \in \mathcal{A}$ and distance  $(2m+1)^{1/p}$ with $s$. Therefore, $\tau_{\mathcal{A}}(S)$ is at distance at most $(3m)^{1/p}$ from every point in $P$.

On the other hand, if for every $S$ in $\mathcal{A}$ there exists $T$ in $\mathcal{B}$ such that $S$ and $T$ are disjoint, then we show that for any point $x$ in $P$ there is a point $y$ in $P$ such that $\|x-y\|_p\ge (3m+1)^{1/p}$. Suppose $x:=\tau_{\mathcal{B}}(T)$ for some $T\in {\mathcal{B}}$ then we have $x$ is at distance $(3m+2)^{1/p}$ from $s$. Similarly if $x:=s$ then it is at distance $(3m+2)^{1/p}$ from every $\tau_{\mathcal{B}}(T)$ for all $T\in {\mathcal{B}}$. Finally, if  $x:=\tau_{\mathcal{A}}(S)$ for some $S\in {\mathcal{A}}$ then from the soundness assumption we have that there exists $T$ in $\mathcal{B}$ such that $S$ and $T$ are disjoint. Thus, $x$ is at distance $(3m+1)^{1/p}$ from $\tau_{\mathcal{B}}(T)$.
\end{proof} 
 
 \begin{remark}\label{rem:inf}
 For the $\ell_{\infty}$-metric, we can solve Discrete 1-center problem in $O(n d^2)$ time as follows. Given input point-set $P$, for every coordinate $i\in [d]$, determine a farthest pair of points $(a_i,b_i)$ in the point-set when restricted to that coordinate. Note that discrete 1-center cost of $P$ is equal to the cost of the discrete 1-center of the point-set $\{a_1,...,a_d,b_1,....,b_d\}$ when the center can be picked anywhere in $P$.  Thus we can solve discrete 1-center in the $\ell_{\infty}$-metric in $O(n d^2)$ time, which is near linear time as long as $d = n^{o(1)}$.
 \end{remark}
 
 The quadratic lower bound for 1-center in Ulam metric follows from the below lemma.

\begin{lemma}\label{lem:hamtoulam}
Let $\Pi_d$ denote the set of all permutations over $[d]$.
For every $d\in\mathbb{N}$, there is a function $\eta:\{0,1\}^d\to \Pi_{2d}$,  such that for all $a,b\in\{0,1\}^d$ the following holds:
\[
\ed(\eta(a),\eta(b))= 2\cdot \|a-b\|_0.
\]
Moreover, for any $a\in \{0,1\}^d$, $\eta(a)$ can be computed in $O(d)$ time.
\end{lemma}
\begin{proof}
Let $a\in\{0,1\}^d$. We define $\eta(a)$ as follows:
\[
\forall i\in [2d],\ \eta(a)[i]=\begin{cases}
i &\text{ if }i=2k-1\text{ and }a_k=0\text{ for some }k\in \mathbb{N}\\
i &\text{ if }i=2k\text{ and }a_k=0\text{ for some }k\in \mathbb{N}\\
i+1 &\text{ if }i=2k-1\text{ and }a_k=1\text{ for some }k\in \mathbb{N}\\
i-1 &\text{ if }i=2k\text{ and }a_k=1\text{ for some }k\in \mathbb{N}
\end{cases}
\]

Fix some $k\in [d]$ and $a,b\in\{0,1\}^d$. If $a_k=b_k$ then notice that $\eta(a)[2k]=\eta(b)[2k]$ and $\eta(a)[2k-1]=\eta(b)[2k-1]$. If $a_k\neq b_k$ then $\eta(a)[2k]=\eta(b)[2k-1]$ and $\eta(a)[2k-1]=\eta(b)[2k]$. Since the characters do not repeat, we have that the optimal distance is obtained by swapping which amounts to two edit operations. 
\end{proof}
 
 \begin{corollary}[Subquadratic Hardness of 1-center in Ulam metric]\label{cor:ulamsubquad}
 Assuming \HSC, for every $\eps>0$ there exists $c>1$ such that no algorithm running in time $n^{2-\eps}$ can given as input a set $P$ of $n$ many permutations of $[d]$, solve the discrete 1-center problem in Ulam metric, where $|P|=n$ and $d=c\log n$.
  \end{corollary}

  The quadratic lower bound for 1-center in Edit metric follows from the below lemma.

\begin{lemma}\label{lem:hamtoedit}
For every $d\in\mathbb{N}$, there is a function $\eta:\{0,1\}^d\to \{0,1\}^{d'}$,  such that for all $a,b\in\{0,1\}^d$ the following holds:
\[
\ed(\eta(a),\eta(b))= \|a-b\|_0.
\]
Moreover, for any $a\in \{0,1\}^d$, $\eta(a)$ can be computed in $O(d \log d)$ time.
\end{lemma}
\begin{proof}
	Let $l_1, l_2, \ldots, l_d$ be $d$ strings of length $10 \log d$ each made by realizing $10 \log d$ $0/1$ bits uniformly at random. It follows that with high probability, the hamming distance as well as the edit distance of each pair $l_i, l_j$ ($i \neq j$) is $\Omega(\log d)$~\cite{kiwi2009speculated}. For a string $a$, we define $\eta(a)$ in the following way: we make a string of size $d (10 \log d + 1)$ which consists of $d$ consecutive blocks. Each block $i$ starts with $a_i$ and is followed by $l_i$. By putting all the blocks next to each other we obtain a string of size $d (10 \log d + 1)$ which we denote by $\eta(a)$. We prove in the following that $\ed(\eta(a),\eta(b))= \|a-b\|_0$ holds for each pair of strings $a$ and $b$.

$\ed(\eta(a),\eta(b)) \leq  \|a-b\|_0$ immediately follows from the fact that by only toggling the first characters of some block of $\eta(a)$ we can turn $\eta(a)$ into $\eta(b)$ and this transformation only costs $ \|a-b\|_0$. Note that we only toggle the first characters of the blocks whose corresponding characters in $a$ and $b$ are not the same.

Now, assume for the sake of contradiction that $\ed(\eta(a),\eta(b)) <  \|a-b\|_0$ holds. This implies that for at least $d-\|a-b\|_0 + 1$ many blocks of $a$, the transformation cost is 0. In other words, for each of these blocks, there is a substring of length $10 \log d +1$ in $\eta(b)$ which is completely the same as that block. Since the blocks are generated randomly, this can only happen if for some $i$, the $i$'th block of $\eta(a)$ is transformed into the $i$'th block of $\eta(b)$ and $a_i = b_i$. Thus, this implies that for at least $d - \|a-b\|_0 + 1$ different values of $i$ we have $a_i = b_i$ which is contradiction.
\end{proof}
\begin{corollary}[Subquadratic Hardness of 1-center in Edit metric]\label{cor:editsubquad}
Assuming \HSC, for every $\eps>0$ there exists $c>1$ such that no algorithm running in time $n^{2-\eps}$ can given as input a point-set $P\subseteq \{0,1\}^d$ solve the discrete 1-center problem in edit metric, where $|P|=n$ and $d=c\log n\log\log n$.
\end{corollary}

Next we prove much higher lower bounds for the Edit metric when $d$ is larger.

\subsection{Subquartic Lower Bound for 1-center when $d=n$ in Edit metric}\label{sec:edit}

We now present our lower bound under Quantified SETH which offers a conceptual novelty since as discussed in Section~\ref{sec:prelim} it is the first time (to our knowledge) that more than two quantifier alternations are utilized.

\begin{theorem}[Subquartic Hardness of 1-center in Edit metric]\label{thm:Edit}
Assuming Quantified SETH, for every $\eps>0$ no algorithm running in time $n^{4-\eps}$ can given as input a point-set $P\subseteq \{0,1\}^n$ solve the discrete 1-center problem in edit metric, where $|P|=n$.
\end{theorem}
\begin{proof}
Let us first reduce to the 1-center problem with facilities where there are two sets of binary strings, a set of clients $C$ and a set of facilities $F$ and the goal is to decide if there is a string in $F$ that has ED at most $\tau$ to all strings in $C$.
Given an instance $\mathcal{A},\mathcal{B},\mathcal{C},\mathcal{D}$ of \EFEEOV we construct $C$ and $F$ as follows.

First we will use the following lemma that follows from the existing reductions from OVH to ED \cite{BI15,BK15} (the latter reference gets the alphabet size down to $2$).

\begin{lemma}[\cite{BK15}]
There are two linear time algorithms such that: each algorithm takes a set ($A$ or $B$) of $n$ binary vectors of length $d$ and constructs (independently of the other) a binary string ($s_A$ or $s_B$)  of length $O(nd)$ with the following property for a fixed value $\tau$ that only depends on $n,d$:
$ED(s_A,s_B) < \tau$ if there is a pair of orthogonal vectors $v_A \in A, v_B \in B$ and $ED(s_A,s_B) \geq \tau$ otherwise.
\end{lemma}

For a set $X \subseteq [d] $ let $v(X) \in \{0,1\}^d$ be the natural encoding of the set as a binary vector where $v(X)[i]=1$ iff $i \in X$.
Note that two vectors are orthogonal iff the two corresponding sets are disjoint.
 
Now, for each set $S_A \in \mathcal{A}$ define the set of $n$ vectors $A = \{ v(X) \mid S_C \in \mathcal{C}, X=S_A \cap S_C \}$ representing the $n$ sets that result from intersecting $S_A$ with any set in $\mathcal{C}$.
Similarly, for each set $S_B \in \mathcal{B}$ define the set of $n$ vectors $B = \{ v(X) \mid S_D \in \mathcal{D}, X=S_B \cap S_D \}$.

It follows that there is an orthogonal pair $v_A \in A, v_B \in B$ iff there exist $S_C \in \mathcal{C}$ and $S_D \in \mathcal{D}$ such that $S_A \cap S_B \cap S_C \cap S_D = \emptyset$.
Therefore, if we use the algorithms in the above lemma to encode each set $A$ with a string $s_A$ and add it to the set of facilities $F$, and also encode each set $B$ with a string $s_B$ and add it into the set of clients $C$ we get the reduction we are after: By the definition of the \EFEEOV problem, there is a string $s_A \in F$ such that for all strings $s_B \in C$ we have $ED(s_A,s_B)< \tau$ if and only if the given \EFEEOV instance is a yes-instance.

Finally, we reduce to the basic 1-center problem (without facilities). Suppose that the strings in $F,C$ have length $m$.
We simply construct an instance $P\subseteq \{0,1\}^{4m}$ of 1-center as follows. 

\[
P:=\{1^m \circ f \circ 0^{2m}\mid f\in F\}\cup \{1^m \circ c\circ 1^{2m}\mid c\in C\}\cup \{0^{4m}\}.
\]

The following simple facts about the ED of the transformed strings show that the optimal center in $P$ must be from $\{1^m \circ f \circ 0^{2m}\mid f\in F\}$ and its cost would be smaller than $2m+\tau$ iff the original \EFEEOV instance is a yes-instance.

\begin{claim}
\label{cl:ED}
Let $x,y$ be two binary strings of length $m$ with ED exactly $t$.
\begin{itemize}
\item $ED(1^m \circ x \circ 0^{2m}, 1^m \circ y \circ 0^{2m}) \leq m$.
\item $ED(1^m \circ x \circ 0^{2m}, 0^{4m}) \leq 2m$.
\item $ED(1^m \circ x \circ 0^{2m}, 1^m \circ y \circ 1^{2m}) = 2m+t$.
\item $ED(1^m \circ x \circ 1^{2m}, 0^{4m}) \geq 3m$.
\end{itemize}
\end{claim}

The first and second items follow from the straightforward alignment of the strings.
The fourth item follows because $ED(0^{\ell},1^{\ell})=\ell$.
The third item requires a bit more care.
First, to see that the ED is at most $2m+t$ consider the alignment that maps $x$ to $y$ optimally at cost $t$ and then maps the other parts in the straightforward way at cost $2m$.
Now suppose for contradiction that there was a better alignment. This alignment must match one of the new letters (from the transformation) to $x$ or $y$; otherwise it would yield an alignment between $x,y$ at cost smaller than $t$. But any alignment that matches the $1$ letters on the left to $x$ or $y$ can be corrected so that the $1^m$ parts on the left are matched to each other, without affecting the cost.
Similarly, any matching between the letters on the right to $x$ or $y$ can be corrected without increasing the cost. Suppose that a $0$ from the right is matched to $y$. This implies that one of the $1$'s to the right of $y$ must be deleted (because there are no longer enough $0$'s in the other string to get substituted with all of them), and a corrected alignment that instead substitutes the $0$ with a $1$ (reducing the number of such deletions by one) and leaves the mate in $y$ unmatched does not have a higher cost. We refer the reader to \cite{BK15} for more formal proofs of such claims.
\end{proof}

\begin{remark}\label{rem:1med}
The above reduction to Edit also work for the 1-median problem but with two key differences.
The first and main difference is that, since we take the sum instead of the max, the cost in the objective may now be affected by non-orthogonal pairs and it is no longer sufficient to have gadgets that give distance $< \tau$ or $\geq \tau$ depending on the orthogonality. Instead, we need gadgets that guarantee that the distance is either $< \tau$ or exactly $\tau$. 
Fortunately, such requirements can be accomplished, see e.g. Theorem~$4$ in \cite{BI15}.
The second difference is that we do not need the $\forall$ quantifier in the starting assumption; the sum is powerful enough to support the (standard) $\exists \exists \exists \exists$ structure type.
Therefore, the lower bounds for 1-median can be based on the standard SETH rather than the Quantified-SETH.
\end{remark}

\subsection{Subquadratic Lower Bounds for 1-center in Low dimensional Euclidean space}\label{sec:euclid}

In this subsection, we show that  an algorithm with subquadratic running time does not exist in  the low dimensional Euclidean metric for the 1-center problem. Our proof essentially adopts ideas developed in \cite{W18,C20}. We note that this result is surprising as there is a near linear time algorithm for 1-center in the low dimensional $\ell_1$-metric.

\begin{remark}\label{rem:ell1}
 For the $\ell_1$-metric, we can solve Discrete 1-center problem in ${O}(n2^{2d})$ time by using the isometric embedding of the  $\ell_1$-metric to the  $\ell_{\infty}$-metric \cite{iso}, and then noting Remark~\ref{rem:inf}.
 \end{remark}

\begin{theorem}\label{thm:euclidean}
Assuming \HSC, there exists a constant $\eta>1$ such that for every $\varepsilon>0$, no algorithm running in time $n^{2-\varepsilon}$ can given as input a point-set $P\subseteq \mathbb{R}^{d}$  and a positive real $\alpha$ output all points in $P$ whose 1-center cost in the Euclidean metric is at most $\alpha$, where $|P|=n$, $d=\eta^{\logstar n}$, and representing each vector requires at most $\tilde{O}(\log n)$ bits. 
\end{theorem} 
\begin{proof}[Proof of Theorem~\ref{thm:euclidean}]
We prove the theorem statement by contradiction. Suppose for some $\eps>0$  there is an algorithm $\mathcal{T}$ running in time $n^{2-\eps}$ that can given as input a point-set $P\subseteq \mathbb{R}^d$ and a positive real $\alpha$ output all points in $P$ whose 1-center cost in the Euclidean metric is at most $\alpha$, where $|P|=n$, $d=\eta^{\logstar n}$, and each vector is of at most $k\log n$ bit entries (for some constant integer $k$).

Let $(\mathcal{A}:=(S_1,\ldots ,S_n),\mathcal{B}:=(T_1,\ldots ,T_n),U)$ be an instance arising from \HSC, where $|U|=c\log n$. We think of each set in $\mathcal{A}$ and $\mathcal{B}$ as its characteristic vector in $\{0,1\}^{c\log n}$.  We show how we can decide this instance in $n^{2-\frac{\varepsilon}{2}}$ time using $\mathcal{T}$, thus contradicting \HSC.
 
We need the following theorem from Chen \cite{C20}.
\begin{theorem}[Chen \cite{C20}]\label{thm:Chen}
Let $b, \ell$ be two sufficiently large integers. There is a reduction $\psi_{b,\ell} : \{0,1\}^{b \cdot \ell} \to \mathbb{Z}^{\ell}$ and a set $V_{b,\ell} \subseteq \mathbb{Z}$, such that for every $x,y \in \{0,1\}^{b \cdot \ell}$,
	
	\[
	x \cdot y = 0 \Leftrightarrow \psi_{b,\ell}(x) \cdot \psi_{b,\ell}(y) \in V_{b,\ell}
	\]
	and 
	\[
	0 \le \psi_{b,\ell}(x)_i < {\ell}^{6^{\logstar(b)} \cdot b}
	\]
	for all possible $x$ and $i \in [\ell]$.
	Moreover, the computation of $\psi_{b,\ell}(x)$ takes $\operatorname*{poly}(b \cdot \ell)$ time, and the set $V_{b,\ell}$ can be constructed in $O\left(\ell^{O(6^{\log^*(b)} \cdot b)} \cdot \operatorname{poly}(b \cdot \ell) \right)$ time.
\end{theorem}

We use the above theorem with  $\ell=7^{\logstar n}$ and $b=|U|/ \ell$. 
Note that if $\ell=7^{\logstar n}$ then $\log \left({\ell}^{6^{\logstar(b)} \cdot b}\right)=o(\log n)$. 
All of the below construction details appears in \cite{W18,C20} and we skip many of the calculations and claim proofs hereafter. Our contributions are mainly in using these previously known constructions in a new way to prove the theorem statement. 
In particular, for every $t\in V_{b,\ell}$ we create an instance $(P_t\subseteq \mathbb{R}^{(\ell+1)^2+3},\alpha:=\sqrt{2n^5-1})$ of 1-center as follows.

For every\footnote{Recall that we think of $S_i$ and $T_j$ through their characteristic vector.}  $S_i\in \mathcal{A}$ (resp.\ $T_j\in\mathcal{B}$) we first define a point $p_i^t\in \mathbb{Z}^{\ell+1}$ (resp.\ $q_j^t\in \mathbb{Z}^{\ell+1}$) as follows:
\[p_i^t:=(\psi_{b,\ell}(S_i),t)\ (\text{resp. }q_j^t:=(\psi_{b,\ell}(T_j),-1)).
\]

It is then easy to verify that $S_i\cap T_j=\emptyset$ if and only if there exists some $t\in V_{b,\ell}$ such that $\langle p_i^t,q_j^t\rangle=0$. Next for every $p_i^t\in \mathbb{Z}^{\ell+1}$ (resp.\ $q_j^t\in \mathbb{Z}^{\ell+1}$) we define $\tilde{p}_i^t\in \mathbb{Z}^{(\ell+1)^2}$ (resp.\ $\tilde{q}_j^t\in \mathbb{Z}^{(\ell+1)^2}$) as follows:
 \[
\forall a,b\in[\ell+1],\ \tilde{p}_i^t(a,b):= p_i^t(a)\cdot p_i^t(b) \ (\text{resp. }\tilde{q}_j^t(a,b):= - q_j^t(a)\cdot q_j^t(b)).
\]

It is  then straightforward to verify that  $\langle p_i^t,q_j^t\rangle=0$ if and only if $\langle \tilde{p}_i^t,\tilde{q}_j^t\rangle\ge 0$.

Finally, we have our pointset $P_t\in\mathbb{R}^{(\ell+1)^2+3}$ defined as follows:
\[
P_t:=\underbrace{\left\{\left(\tilde{p}_i^t,\sqrt{n^5-\|\tilde{p}_i^t\|_2^2},0,0\right)\bigg| S_i\in \mathcal{A}\right\}}_{P_t^{\mathcal{A}}}\bigcup\underbrace{\left\{\left(-\tilde{q}_j^t,0,\sqrt{n^5-\|\tilde{q}_j^t\|_2^2},\sqrt{n^5}\right)\bigg| T_j\in \mathcal{B}\right\}}_{P_t^{\mathcal{B}}}\cup\left\{\vec{0}\right\},	
\]
where $\vec{0}=(0,0,\ldots ,0)$.

It can then be verified that   $\langle \tilde{p}_i^t,\tilde{q}_j^t\rangle\ge 0$ if and only if the distance between $\left(\tilde{p}_i^t,\sqrt{n^5-\|\tilde{p}_i^t\|_2^2},0\right)$ and $\left(-\tilde{q}_j^t,0,\sqrt{n^5-\|\tilde{q}_j^t\|_2^2}\right)$ is at least $\sqrt{2n^5}$; otherwise their distance is at most $\sqrt{2n^5-1}$. Also note that any pair of points in ${P_t^{\mathcal{A}}}$ or any pair of points in ${P_t^{\mathcal{B}}}$ are at distance at most $\sqrt{2n^5-1}$ from each other. Finallt, note that the distance between any point in ${P_t^{\mathcal{B}}}$ and $\vec{0}$ is exactly $\sqrt{2n^5}$ and the distance between any point in ${P_t^{\mathcal{A}}}$ and $\vec{0}$ is exactly $\sqrt{n^5}$.

We run $\mathcal{T}$ on $(P_t,\alpha:=\sqrt{2n^5-1})$ for every $t\in V_{b,\ell}$. 
Let $\mathcal{O}_t\subseteq P_t$ be the output of running $\mathcal{T}$ on $(P_t,\alpha)$. In other words for every $t\in V_{b,\ell}$ and every $p\in\mathcal{O}_t$ we have that for every $p'\in P_t$, $\|p-p'\|_2\le \sqrt{2n^5-1}$.

We claim that there exists $S$ in $\mathcal{A}$ such that it intersects with every subset $T$ in $\mathcal{B}$ if and only if there exists $i\in[n]$ such that for all $t\in V_{b,\ell}$, we have $\left(\tilde{p}_i^t,\sqrt{n^5-\|\tilde{p}_i^t\|_2^2},0\right)\in \mathcal{O}_t$. 

Suppose there exists $S_{i^*}$ in $\mathcal{A}$ such that it intersects with every subset $T$ in $\mathcal{B}$. Fix $t\in V_{b,\ell}$. We have that $\left(\tilde{p}_{i^*}^t,\sqrt{n^5-\|\tilde{p}_{i^*}^t\|_2^2},0\right)$ is at distance at most $\sqrt{2n^5-1}$ from every other point in $P_t^{\mathcal{A}}$ just from construction. Suppose there is $\left(-\tilde{q}_j^t,0,\sqrt{n^5-\|\tilde{q}_j^t\|_2^2}\right)\in P_t^{\mathcal{B}}$ such that their distance is greater than $\sqrt{2n^5-1}$ then from the construction, their distance must be $\sqrt{2n^5}$, which implies that $S_{i^*}\cap T_j=\emptyset$, a contradiction.

On the other hand, if for every $S$ in $\mathcal{A}$ there exists $T$ in $\mathcal{B}$ such that $S$ and $T$ are disjoint then we show that for any point $\left(\tilde{p}_{i}^t,\sqrt{n^5-\|\tilde{p}_{i}^t\|_2^2},0\right)$ there exists $t\in V_{b,\ell}$ such that $\left(\tilde{p}_{i}^t,\sqrt{n^5-\|\tilde{p}_{i}^t\|_2^2},0\right)\notin \mathcal{O}_t$. Fix $i\in [n]$. Let $T_j\in \mathcal{B}$ such that $S_i\cap T_j=\emptyset$. Let $t^*:=\psi_{b,\ell}(S_i)\cdot \psi_{b,\ell}(T_j)$.  From Theorem~\ref{thm:Chen} we have that $t^*\in V_{b,\ell}$. Thus $\left(\tilde{p}_{i}^{t^*},\sqrt{n^5-\|\tilde{p}_{i}^{t^*}\|_2^2},0\right)$ and $\left(-\tilde{q}_j^{t^*},0,\sqrt{n^5-\|\tilde{q}_j^{t^*}\|_2^2}\right)$ in $P_{t^*}$ are at distance at least $\sqrt{2n^5}$ and thus $\left(\tilde{p}_{i}^{t^*},\sqrt{n^5-\|\tilde{p}_{i}^{t^*}\|_2^2},0\right)\notin \mathcal{O}_{t^*}$.

Finally, note that the total run time was $O(n^{2-\varepsilon}\cdot |V_{b,\ell}|)=O(n^{2-\varepsilon}\log n)<n^{2-\frac{\varepsilon}{2}}$.
\end{proof}
\section{An $n^{2.5}$ time $1+\epsilon$ Approximation Algorithm for 1-Center in Ulam Metric when $d=n$}
\label{sec:ulam-up}
In this section, we consider the 1-center problem under Ulam metric. More precisely, we consider a problem where $n$ strings $s_1, s_2, \ldots, s_n$ are given as input and our goal is to find a string $s_k$ such that the maximum distance of $s_k$ from the rest of the strings is minimized. Our focus here is on the Ulam metric.

We assume throughout this section that the length of all strings is equal to $d$. Our algorithm for this case is two-fold. Let $o$ be the value of the solution (i.e., the maximum distance of the center of the strings to the rest of the strings is exactly equal to $o$). If $o$ is lower bounded by $\sqrt{d}$, previous work on Ulam distance gives us a $1+\epsilon$ approximate solution for center in the following way: We iterate over all pairs of strings and each time we estimate their Ulam distance via the algorithm of Naumovitz, Saks, and Seshadhri~\cite{naumovitz2017accurate} for approximating the Ulam distance of each pair. When the Ulam distance of two strings is equal to $u$, their algorithm takes time $\tilde O_{\epsilon}(d/u + \sqrt{d})$ to $1+\epsilon$ approximate the solution. Thus, we only run their algorithm up to a runtime of $\tilde O_{\epsilon}(\sqrt{d})$ to either obtain a $1+\epsilon$ approximate solution for the Ulam distance or verify that the Ulam distance is smaller than $\sqrt{d}$. It follows that if $o \geq \sqrt{d}$ this information is enough for us to approximate the 1-center problem within a factor $1+\epsilon$ and the runtime of the algorithm is bounded by $\tilde O_{\epsilon}(n^2\sqrt{d})$. Thus, it only remains to design an algorithm for the low-distance regime.

From here on, we assume that $o \leq \sqrt{d}$. In this case, we take an arbitrary string (say $s_1$) and compute the Ulam distance of that string to all the other strings. In addition to this, we also keep track of the changes that convert $s_1$ into all the strings. It follows that since $o \leq \sqrt{d}$, the distance of $s_1$ to all the strings is bounded by at most $2\sqrt{d}$. Thus, via the transformations we compute in this step, we would be able to make a transformation from any $s_i$ to any $s_j$ with at most $4\sqrt{d}$ operations (we can combine the transformation from $s_1$ to $s_i$ and the transformation from $s_1$ to $s_j$). Using this information, we can determine the exact Ulam distance of every string $s_i$ to every string $s_j$ in the following way:

We start with the non-optimal transformation from $s_i$ to $s_j$ that uses at most $4\sqrt{d}$ operations. We then split the characters from $[1, \ldots, d]$ into buckets such that in each bucket all the characters are next to each other and they appear in the same order in the two strings. To be more precise, consider the following procedure: color each character of $s_i$ and $s_j$ which is touched in the transformation (deleted, added, or changed) in red and the rest blue. Each set of consecutive blue characters and each single red character makes a bucket.
 It follows that because there is a transformation from $s_i$ to $s_j$ with at most $4\sqrt{d}$ operations, the total number of buckets would be bounded by $O(\sqrt{d})$. Moreover, there exists an optimal transformation wherein either all characters of each buckets are deleted/inserted or all characters of each bucket remain intact. This implies that we can compress the two strings into smaller strings by replacing each bucket with a single character. The insertion and deletion of these special characters then has a cost proportional to the size of the bucket. This way, the size of the two strings would be bounded by $O(\sqrt{d})$ and thus we can compute the Ulam distance of the two strings in time $\tilde O(\sqrt{d})$. Therefore, we can compute the center of the strings in time $\tilde O_{\epsilon}(nd + n^2\sqrt{d})$.

	\begin{algorithm}[tbh]
	\KwData{$s_1, s_2, \ldots, s_n$}
	\KwResult{1-center}
	$o \leftarrow \infty$\;
	\For{$i \leftarrow 1$ to $n$}{
		$mx \leftarrow -1$\;
		\For {$j \leftarrow 1$ to $n$}{
			Run ~\cite{naumovitz2017accurate} on $s_i$ and $s_j$ up to $\tilde O_{\epsilon}(\sqrt{d})$ steps\;
			\If {the algorithm terminates}{
				$mx \leftarrow \max\{mx, \text{the output of the algorithm}\}$\;
			}
		}
		$o \leftarrow \min\{o, mx\}$\;
	}
	\If{$o \neq -1$}{
		\Return o\;
	}\Else{
		$o \leftarrow \infty$\;
		\For{$i \leftarrow 1$ to $n$}{
			$tr_i \leftarrow $ optimal transformation between $s_1$ and $s_i$\;
		}
		\For{$i \leftarrow 1$ to $n$}{
			$mx \leftarrow -1$\;
			\For{$j \leftarrow 1$ to $n$}{
				$(s^*_i, s^*_j) \leftarrow $ compressed versions of $(s_i, s_j)$ based on $tr_i$ and $tr_j$\;
				$mx \leftarrow \max\{mx, \text{the Ulam distance of } s^*_i \text{ and } s^*_j\}$\;
			}
			$o \leftarrow \min\{o, mx\}$\;
		}
		\Return $o$\;
	}
	\caption{1-center of $n$ strings under Ulam metric.}
\end{algorithm}

\upperl*

\begin{proof}[Proof of Theorem~\ref{thm:ulam:upper}]
 The outline of the algorithm along with its runtime analysis is given earlier. Here we prove that the approximation factor of the algorithm is bounded by $1+\epsilon$. In case $o \geq \sqrt{d}$, we use the $1+\epsilon$ approximation algorithm of Naumovitz, Saks, and Seshadhri~\cite{naumovitz2017accurate} for each pair of strings up to a runtime of $\tilde O_{\epsilon}(\sqrt{d})$. If the algorithm gives an estimation before we terminate it, we take the value into account when determining the maximum distance for the strings involved. It follows that since $o \geq \sqrt{d}$, then for each string $s_x$, there is one string $s_y$ whose distance to $s_x$ is at least $\sqrt{d}$ and thus the maximum distance we determine for each string is a $1+\epsilon$ approximation of the optimal value. 
 
 Next, we show that in case $o \leq \sqrt{d}$, our algorithm determines the Ulam distance of each pair exactly and thus solves the 1-center problem correctly. In order to determine the Ulam distance between $s_i$ and $s_j$, we begin with a transformation of size at most $4\sqrt{d}$ between the two strings. We then mark all the characters that are either deleted or inserted in this transformation and all the characters that are next to these characters. We then compress the two strings in the following way: each marked character becomes a single character with the same value. Each maximal interval of unmarked characters that are next to each other also become a single character whose value represents the entire interval. Therefore, the compressed strings have $O(\sqrt{d})$ characters each. It follows that the Ulam distance of the compressed strings is exactly equal to the Ulam distance of the original strings. Moreover, even though for the compressed strings the operations have arbitrary costs, we can still solve Ulam distance in time proportional to the length of the strings which results in an algorithm with runtime $\tilde O(\sqrt{d})$ for computing Ulam distance between each pair.
 \end{proof}
\section{Hardness of Approximation of 1-center in String and $\ell_p$-metrics}\label{sec:inapprox}

In this section, we prove hardness of approximation results for the 1-center problem but first we need some error correcting codes background. 

\paragraph*{Error-Correcting Codes.}
An error correcting code $C$ over alphabet $\Sigma$ is a function $C: \Sigma^m \to \Sigma^{\ell}$ where $m$ and $\ell$ are positive integers which are referred to as the {\em message length} and {\em block length} of $C$ respectively. Intuitively, the function $C$ encodes an original message of length $m$ to an encoded message of length $\ell$.
The {\em rate} of a code $\rho(C)$ is defined as the ratio between its message length and its block length, i.e., $\rho(C) = \nicefrac{m}{\ell}$. 
The {\em relative distance} of a code, denoted by $\delta(C)$, is defined as $\underset{x \ne y \in \Sigma^m}{\min} \delta(C(x), C(y))$ where $\delta(C(x), C(y))$ is the {\em relative Hamming distance} between $C(x)$ and $C(y)$, i.e., the fraction of coordinates on which $C(x)$ and $C(y)$ disagree.

In this paper,
we require our codes to have some special algebraic properties which have been shown to be present in algebraic geometric codes \cite{GS96}. First, we will introduce a couple of additional definitions. 

\begin{definition}[Systematicity]
Given $s \in \mathbb N$, a code $C:\Sigma^m\to \Sigma^{\ell}$ is {\em $s$-systematic} if there exists a size-$s$ subset of $[\ell]$, which for convenience we identify with $[s]$, such that for every $x \in \Sigma^{s}$ there exists $w \in \Sigma^m$ in which $x = C(w)\mid_{[s]}$. 
\end{definition}

\begin{definition}[Degree-$t$ Closure] \label{def:poly}
Let $\Sigma$ be a finite field. 
Given two codes $C:\Sigma^m\to \Sigma^{\ell},C':\Sigma^{m'}\to \Sigma^{\ell}$ and positive integer $t$, we say that $C'$ is a degree-$t$ closure of $C$ if, for every $w_1,\ldots ,w_r\in \Sigma^m$ and $P\in \mathbb F[X_1,\ldots ,X_r]$ of total degree at most $t$, it holds that $\omega:=P(C(w_1),\ldots ,C(w_r))$ is in the range of $C'$, where $\omega\in \Sigma^\ell$ is defined coordinate-wise by the equation $\omega_i:=P(C(w_1)_i,\ldots ,C(w_r)_i)$.
\end{definition}

Below we provide a self-contained statement of the result we need; it follows from Theorem~7 of \cite{SAKSD01}, which
gives an efficient construction of the algebraic geometric codes based on \cite{GS96}'s explicit towers of function fields.

\begin{theorem}[\cite{GS96,SAKSD01}] \label{thm:ag-code}
There is a constant $\lambda>0$ such that for any prime $q\ge 7$, there are two code families $\mathcal{C}^1 = \{C_n^1\}_{n \in \mathbb{N}},\ \mathcal{C}^2 = \{C_n^2\}_{n \in \mathbb{N}}$ such that the following holds for all $n\in\mathbb N$,
\begin{itemize}
\item  $C_n^1$ and $C_n^2$ are $n$-systematic code with alphabet $\mathbb{F}_{q^2}$,
\item  $C_n^1$ and $C_n^2$ have block length less than $\lambda n$.
\item  $C_n^2$ has relative distance $ \geq \nicefrac{1}{2}$,
\item  $C_n^2$ is a degree-$2$ closure of $C_n^1$, and,
\item  Any codeword in $C_n^1$ or $C_n^2$ can be computed in poly($n$) time.
\end{itemize}
\end{theorem}

We  point the interested reader to \cite{KLM19} for a proof sketch of the above theorem.

\begin{theorem}[Subquadratic Inapproximability of 1-center in $\ell_p$-metrics]\label{thm:hardnessapprox}
Let $p\in \mathbb{R}_{\ge 1}\cup\{0\}$.
Assuming \HSC, for every $\eps>0$ there exists $\delta>0$  such that no algorithm running in time $n^{2-\eps}$ can given as input a point-set $P\subseteq \{0,1\}^d$ and a positive real $\alpha$ output all points in $P$ whose 1-center cost in the $\ell_p$-metric is at most $\alpha\cdot (1+\delta)$, where $|P|=n$ and $d=O_{\varepsilon}(\log n)$.
\end{theorem}
\begin{proof}
Fix $p\in \mathbb{R}_{\ge 1}\cup\{0\}$.
We prove the theorem statement by contradiction. Suppose for some $\eps>0$  there is an algorithm $\mathcal{T}$ running in time $n^{2-\eps}$ that can for every $\delta>0$, given as input a point-set $P\subseteq \mathbb{R}^d$  and a positive real $\alpha$ output all points in $P$ whose 1-center cost in the $\ell_p$-metric is at most $\alpha\cdot (1+\delta)$, where $|P|=n$ and $d=O_{\varepsilon}(\log n)$ (dependency on $\varepsilon$ will become clear later in the proof).

Let $(\mathcal{A}:=(S_1,\ldots ,S_n),\mathcal{B}:=(T_1,\ldots ,T_n),U)$ be an instance arising from \HSC, where $|U|=c\log n$. We think of each set in $\mathcal{A}$ and $\mathcal{B}$ as its characteristic vector in $\{0,1\}^{c\log n}$.  We show how we can decide this instance in $n^{2-\frac{\varepsilon}{2}}$ time using $\mathcal{T}$, thus contradicting \HSC.  
 
The construction below is exactly the same as the one suggested by Rubinstein \cite{R18}. We however, use this construction to fit our purposes of proving lower bound for the  1-center problem.

\subsubsection*{Algebrization}
Fix  $T=2c/\varepsilon$. 
Let $q$ be the smallest prime greater than $T$ (i.e., $q<2\cdot T$). 
Let $m:=c\log n$. Let $C_{m/T}^1$ and $C_{m/T}^2$ be the codes guaranteed in Theorem~\ref{thm:ag-code} over $\mathbb{F}_{q^2}$ with block length $\ell\le \lambda m/T$.

Let $\tilde{C}\subseteq C_{m/T}^2$ such that  $\omega\in \tilde{C}$ if and only if  $\omega\mid_{[m/T]}= \vec{0}$ (i.e., $\omega$ has a zero entry in each of the first $m/T$ coordinates). For every $\omega\in\tilde{C}$ we define two functions $\tau_{\mathcal{A}}^\omega,\tau_{\mathcal{B}}^\omega:\{0,1\}^{m}\to \{0,1\}^{r}$, where $r:= q^{4(T+2)} \times \ell$. We can thus index every $i\in [r]$ using elements in $ \mathbb{F}_{q^2}^T\times \mathbb{F}_{q^2}^T\times [\ell]$.  

Fix $x\in \{0,1\}^m$. 
Let $x=(x^1,\ldots ,x^T)\in\{0,1\}^m$ where for all $i\in[T]$ we have $x^i\in\{0,1\}^{m/T}$. 
We define $\tau_{\mathcal{A}}^\omega(x)$ coordinate wise. Fix $\zeta\in[r]$ and we think of $\zeta$ as follows:
 \begin{align}
 \zeta=\left((\mu_1^{\mathcal{A}},\ldots,\mu_{T+2}^{\mathcal{A}}),(\mu_1^{\mathcal{B}},\ldots,\mu_{T+2}^{\mathcal{B}}),j\right)\in  \mathbb{F}_{q^2}^{T+2}\times\mathbb{F}_{q^2}^{T+2}\times [\ell].\label{eq:zeta}
 \end{align}
 
We define $\tau_{\mathcal{A}}^\omega(x)_{\zeta}$ to be 1 if and only if: 
\[\underset{i\in[T+2]}{\sum} \mu_i^{\mathcal{A}} \cdot \mu_i^{\mathcal{B}}=\omega(j)\text{\ \ and \ \ }\forall i\in[T],\  \mu_i^{\mathcal{A}}=C_{m/T}^1(x^i)_j,\ \text{and }\mu_{T+1}^{\mathcal{A}}=0,\ \mu_{T+2}^{\mathcal{A}}=C_{m/T}^1(\mathbbm{1}^{m/T})_j.\]

 Similarly, we define $\tau_{\mathcal{B}}^\omega(x)$ coordinate wise. Fix $\zeta\in[r]$ and we think of $\zeta$ as in \eqref{eq:zeta}.
We define $\tau_{\mathcal{B	}}^\omega(x)_{\zeta}$ to be 1 if and only if:
\[\underset{i\in[T+2]}{\sum} \mu_i^{\mathcal{A}} \cdot \mu_i^{\mathcal{B}}=\omega(j)\text{\ \ and \ \ }\forall i\in[T],\  \mu_i^{\mathcal{B}}=C_{m/T}^1(x^i)_j\ \text{and }\mu_{T+1}^{\mathcal{B}}=C_{m/T}^1(\mathbbm{1}^{m/T})_j,\ \mu_{T+2}^{\mathcal{B}}=0.\]

\subsubsection*{Construction}

For every $S_i\in \mathcal{A}$, we define $s_i^{\omega}:=\tau_{\mathcal{A}}^{\omega}(S_i)$. Further, we define $\tilde{s}_i^{\omega}=(s_i^{\omega},\mathbbm{1}^r-s_i^{\omega},\mathbbm{1}^{2r})\in \{0,1\}^{4r}$. 
For every $T_j\in \mathcal{B}$, we define $t_j^{\omega}:=\tau_{\mathcal{B}}^{\omega}(T_j)$. Further, we define $\tilde{t}_j^{\omega}=(\mathbbm{1}^r-t_j^{\omega},t_j^{\omega},{0}^{2r})\in \{0,1\}^{4r}$. 

We define the point-set $P_{\omega}$ to be $P_{\omega}^{\mathcal{A}}:=\{\tilde{s}_i^{\omega}\mid S_i\in\mathcal{A}\}\cup P_{\omega}^{\mathcal{B}}:=\{\tilde{t}_j^{\omega}\mid T_j\in \mathcal{B}\}\cup\{\mathbbm{1}^{4r}\}$. 
Let $\alpha:=(2q^{4(T+2)}-4q^{2(T+1)})\cdot \ell + 2r +\ell$.  Let $\delta:=1/(4q^{4T}-4q^{2T-2} +1)$.

\subsubsection*{Analysis}

We run $\mathcal{T}$ on $(P_{\omega},\alpha)$ for every $\omega\in \tilde{C}$. 
Let $\mathcal{O}_{\omega}\subseteq P_{\omega}$ be the output of running $\mathcal{T}$ on $(P_{\omega},\alpha)$. In other words for every $\omega\in \tilde{C}$ and every $s\in\mathcal{O}_t$ we have that for every $s'\in P_{\omega}$, $\|s-s'\|_p\le (1+\delta)^{1/p}\cdot \alpha^{1/p}$.

We claim that there exists $S$ in $\mathcal{A}$ such that it intersects with every subset $T$ in $\mathcal{B}$ if and only if there exists $i\in[n]$ such that for all $\omega\in \tilde{C}$, we have $\tilde{s}_i^{\omega}\in \mathcal{O}_\omega$. 

Suppose there exists $S_{i^*}$ in $\mathcal{A}$ such that it intersects with every subset $T$ in $\mathcal{B}$. Fix $\omega\in\tilde{C}$. We have that $\tilde{s}_{i^*}^{\omega}$ is at distance at most $(2r)^{1/p}$ from every other point in $P_{\omega}^{\mathcal{A}}$ just from construction. Suppose there is $\tilde{t}_{j}^{\omega}\in P_{\omega}^{\mathcal{B}}$ such that their distance is greater than $\alpha^{1/p}$ then from the construction, their distance must be at least $(1+\delta)^{1/p}\cdot \alpha^{1/p}$, which implies that $S_{i^*}\cap T_j=\emptyset$, a contradiction.

On the other hand, if for every $S$ in $\mathcal{A}$ there exists $T$ in $\mathcal{B}$ such that $S$ and $T$ are disjoint then we show that for any point $\tilde{s}_{i}^{\omega}$ there exists $\omega\in\tilde{C}$ such that $\tilde{s}_{i}^{\omega}\notin \mathcal{O}_{\omega}$. Fix $i\in [n]$. Let $T_j\in \mathcal{B}$ such that $S_i\cap T_j=\emptyset$. Let $\omega^*:=\sum_{e\in [T]}C_{m/T}^1(S_i^e)\cdot C_{m/T}^1(T_j^e)$.  From Theorem~\ref{thm:ag-code} we have that $\omega^*\in \tilde{C}$. Thus $\tilde{s}_{i}^{\omega^*}$ and $\tilde{t}_j^{\omega^*}$ in $P_{\omega^*}$ are at distance at least $(1+\delta)^{1/p}\cdot \alpha^{1/p}$ and thus $\tilde{s}_i^{\omega^*}\notin \mathcal{O}_{\omega^*}$. Also, note that for all $\omega\in \tilde{C}$, we have that every point in $P_{\omega}^{\mathcal{B}}$ is at distance at least $(3r)^{1/p}$ from $\mathbbm{1}^{4r}$.

Finally, note that the total run time was $O(n^{2-\varepsilon}\cdot |\tilde{C}|)=O(n^{2-\varepsilon+\frac{c}{T}})<n^{2-\frac{\varepsilon}{2}}$.
\end{proof}

We remark that the above construction is the same as the one in \cite{R18} albeit for a different problem (nearest neighbors problem).

Theorem~\ref{thm:hardnessapprox} readily extend to the edit metric from the below statement and to the Ulam metric from Lemma~\ref{lem:hamtoulam}.

\begin{lemma}[Rubinstein \cite{R18}]\label{lem:edit}
For large enough $d\in\mathbb{N}$, there is a function $\eta:\{0,1\}^d\to \{0,1\}^{d'}$, where $d'=O(d\log d)$, such that for all $a,b\in\{0,1\}^d$ the following holds for some constant $\lambda>0$:
\[\left|\ed(\eta(a),\eta(b))-\lambda\cdot \log d\cdot \|a-b\|_0\right|= o(d').
\]
Moreover, for any $a\in \{0,1\}^d$, $\eta(a)$ can be computed in $2^{o(d)}$ time.
\end{lemma}

\bibliographystyle{alpha}
\bibliography{references.bib}

\newcommand{\etalchar}[1]{$^{#1}$}
\begin{thebibliography}{LMW02b}

\bibitem[ABHS20]{ABHS20}
Amir Abboud, Karl Bringmann, Danny Hermelin, and Dvir Shabtay.
\newblock Scheduling lower bounds via {AND} subset sum.
\newblock In Artur Czumaj, Anuj Dawar, and Emanuela Merelli, editors, {\em 47th
  International Colloquium on Automata, Languages, and Programming, {ICALP}
  2020, July 8-11, 2020, Saarbr{\"{u}}cken, Germany (Virtual Conference)},
  volume 168 of {\em LIPIcs}, pages 4:1--4:15. Schloss Dagstuhl -
  Leibniz-Zentrum f{\"{u}}r Informatik, 2020.

\bibitem[ARW17]{ARW17}
Amir Abboud, Aviad Rubinstein, and R.~Ryan Williams.
\newblock Distributed {PCP} theorems for hardness of approximation in {P}.
\newblock In {\em 58th {IEEE} Annual Symposium on Foundations of Computer
  Science, {FOCS} 2017, Berkeley, CA, USA, October 15-17, 2017}, pages 25--36,
  2017.

\bibitem[AWW16]{AWW16}
Amir Abboud, Virginia~Vassilevska Williams, and Joshua~R. Wang.
\newblock Approximation and fixed parameter subquadratic algorithms for radius
  and diameter in sparse graphs.
\newblock In Robert Krauthgamer, editor, {\em Proceedings of the Twenty-Seventh
  Annual {ACM-SIAM} Symposium on Discrete Algorithms, {SODA} 2016, Arlington,
  VA, USA, January 10-12, 2016}, pages 377--391. {SIAM}, 2016.

\bibitem[BC19]{BC19}
Karl Bringmann and Bhaskar~Ray Chaudhury.
\newblock Polyline simplification has cubic complexity.
\newblock In Gill Barequet and Yusu Wang, editors, {\em 35th International
  Symposium on Computational Geometry, SoCG 2019, June 18-21, 2019, Portland,
  Oregon, {USA}}, volume 129 of {\em LIPIcs}, pages 18:1--18:16. Schloss
  Dagstuhl - Leibniz-Zentrum f{\"{u}}r Informatik, 2019.

\bibitem[BI15]{BI15}
Arturs Backurs and Piotr Indyk.
\newblock Edit distance cannot be computed in strongly subquadratic time
  (unless {SETH} is false).
\newblock In Rocco~A. Servedio and Ronitt Rubinfeld, editors, {\em Proceedings
  of the Forty-Seventh Annual {ACM} on Symposium on Theory of Computing, {STOC}
  2015, Portland, OR, USA, June 14-17, 2015}, pages 51--58. {ACM}, 2015.

\bibitem[BK15]{BK15}
Karl Bringmann and Marvin K{\"{u}}nnemann.
\newblock Quadratic conditional lower bounds for string problems and dynamic
  time warping.
\newblock In Venkatesan Guruswami, editor, {\em {IEEE} 56th Annual Symposium on
  Foundations of Computer Science, {FOCS} 2015, Berkeley, CA, USA, 17-20
  October, 2015}, pages 79--97. {IEEE} Computer Society, 2015.

\bibitem[BLPR97]{BLPR97}
Amir Ben{-}Dor, Giuseppe Lancia, Jennifer Perone, and R.~Ravi.
\newblock Banishing bias from consensus sequences.
\newblock In Alberto Apostolico and Jotun Hein, editors, {\em Combinatorial
  Pattern Matching, 8th Annual Symposium, {CPM} 97, Aarhus, Denmark, June 30 -
  July 2, 1997, Proceedings}, volume 1264 of {\em Lecture Notes in Computer
  Science}, pages 247--261. Springer, 1997.

\bibitem[CDK21]{CDK21}
Diptarka Chakraborty, Debarati Das, and Robert Krauthgamer.
\newblock Approximating the median under the ulam metric.
\newblock In D{\'{a}}niel Marx, editor, {\em Proceedings of the 2021 {ACM-SIAM}
  Symposium on Discrete Algorithms, {SODA} 2021, Virtual Conference, January 10
  - 13, 2021}, pages 761--775. {SIAM}, 2021.

\bibitem[CGI{\etalchar{+}}16]{CGIMPS16}
Marco~L. Carmosino, Jiawei Gao, Russell Impagliazzo, Ivan Mihajlin, Ramamohan
  Paturi, and Stefan Schneider.
\newblock Nondeterministic extensions of the strong exponential time hypothesis
  and consequences for non-reducibility.
\newblock In Madhu Sudan, editor, {\em Proceedings of the 2016 {ACM} Conference
  on Innovations in Theoretical Computer Science, Cambridge, MA, USA, January
  14-16, 2016}, pages 261--270. {ACM}, 2016.

\bibitem[Che20]{C20}
Lijie Chen.
\newblock On the hardness of approximate and exact (bichromatic) maximum inner
  product.
\newblock {\em Theory Comput.}, 16:1--50, 2020.

\bibitem[CHW10]{ClarksonHW10}
Kenneth~L. Clarkson, Elad Hazan, and David~P. Woodruff.
\newblock Sublinear optimization for machine learning.
\newblock In {\em 51th Annual {IEEE} Symposium on Foundations of Computer
  Science, {FOCS} 2010, October 23-26, 2010, Las Vegas, Nevada, {USA}}, pages
  449--457. {IEEE} Computer Society, 2010.

\bibitem[CKL21]{DBLP:conf/soda/Cohen-AddadSL21}
Vincent Cohen{-}Addad, {Karthik {C. S.}}, and Euiwoong Lee.
\newblock On approximability of clustering problems without candidate centers.
\newblock In D{\'{a}}niel Marx, editor, {\em Proceedings of the 2021 {ACM-SIAM}
  Symposium on Discrete Algorithms, {SODA} 2021, Virtual Conference, January 10
  - 13, 2021}, pages 2635--2648. {SIAM}, 2021.

\bibitem[CLM{\etalchar{+}}16]{CLMPS16}
Michael~B. Cohen, Yin~Tat Lee, Gary~L. Miller, Jakub Pachocki, and Aaron
  Sidford.
\newblock Geometric median in nearly linear time.
\newblock In Daniel Wichs and Yishay Mansour, editors, {\em Proceedings of the
  48th Annual {ACM} {SIGACT} Symposium on Theory of Computing, {STOC} 2016,
  Cambridge, MA, USA, June 18-21, 2016}, pages 9--21. {ACM}, 2016.

\bibitem[DKL19]{DKL19}
Roee David, {Karthik {C. S.}}, and Bundit Laekhanukit.
\newblock On the complexity of closest pair via polar--pair of point--sets.
\newblock {\em {SIAM} J. Discrete Math.}, 33(1):509--527, 2019.

\bibitem[dlHC00]{HC00}
Colin de~la Higuera and Francisco Casacuberta.
\newblock Topology of strings: Median string is np-complete.
\newblock {\em Theoretical Computer Science}, 230(1):39--48, 2000.

\bibitem[FL97]{DBLP:journals/mst/FrancesL97}
Moti Frances and Ami Litman.
\newblock On covering problems of codes.
\newblock {\em Theory Comput. Syst.}, 30(2):113--119, 1997.

\bibitem[GGJ76]{GGJ76}
M.~R. Garey, Ronald~L. Graham, and David~S. Johnson.
\newblock Some np-complete geometric problems.
\newblock In Ashok~K. Chandra, Detlef Wotschke, Emily~P. Friedman, and
  Michael~A. Harrison, editors, {\em Proceedings of the 8th Annual {ACM}
  Symposium on Theory of Computing, May 3-5, 1976, Hershey, Pennsylvania,
  {USA}}, pages 10--22. {ACM}, 1976.

\bibitem[GJL99]{GJL99}
Leszek Gasieniec, Jesper Jansson, and Andrzej Lingas.
\newblock Efficient approximation algorithms for the hamming center problem.
\newblock In Robert~Endre Tarjan and Tandy~J. Warnow, editors, {\em Proceedings
  of the Tenth Annual {ACM-SIAM} Symposium on Discrete Algorithms, 17-19
  January 1999, Baltimore, Maryland, {USA}}, pages 905--906. {ACM/SIAM}, 1999.

\bibitem[GS96]{GS96}
Arnaldo Garcia and Henning Stichtenoth.
\newblock On the asymptotic behaviour of some towers of function fields over
  finite fields.
\newblock {\em Journal of Number Theory}, 61(2):248 -- 273, 1996.

\bibitem[KLM19]{KLM19}
{Karthik {C. S.}}, Bundit Laekhanukit, and Pasin Manurangsi.
\newblock On the parameterized complexity of approximating dominating set.
\newblock {\em J. {ACM}}, 66(5):33:1--33:38, 2019.

\bibitem[KRW91]{KRW91}
Mauricio Karchmer, Ran Raz, and Avi Wigderson.
\newblock Super-logarithmic depth lower bounds via direct sum in communication
  coplexity.
\newblock In {\em Proceedings of the Sixth Annual Structure in Complexity
  Theory Conference, Chicago, Illinois, USA, June 30 - July 3, 1991}, pages
  299--304. {IEEE} Computer Society, 1991.

\bibitem[KS09]{kiwi2009speculated}
Marcos Kiwi and Jos{\'e} Soto.
\newblock On a speculated relation between chv{\'a}tal--sankoff constants of
  several sequences.
\newblock {\em Combinatorics, Probability and Computing}, 18(4):517--532, 2009.

\bibitem[LLM{\etalchar{+}}03]{LLMWZ03}
J.~Kevin Lanct{\^{o}}t, Ming Li, Bin Ma, Shaojiu Wang, and Louxin Zhang.
\newblock Distinguishing string selection problems.
\newblock {\em Inf. Comput.}, 185(1):41--55, 2003.

\bibitem[LMW02a]{LMW02a}
Ming Li, Bin Ma, and Lusheng Wang.
\newblock Finding similar regions in many sequences.
\newblock {\em J. Comput. Syst. Sci.}, 65(1):73--96, 2002.

\bibitem[LMW02b]{LiMW02}
Ming Li, Bin Ma, and Lusheng Wang.
\newblock On the closest string and substring problems.
\newblock {\em J. {ACM}}, 49(2):157--171, 2002.

\bibitem[Mon08]{iso}
Ashley Montanaro.
\newblock {Metric Embeddings}.
\newblock
  \url{http://people.maths.bris.ac.uk/~csxam/presentations/embeddings.pdf},
  2008.
\newblock [Online; accessed 12-December-2008].

\bibitem[MS20]{MS20}
Michael Mitzenmacher and Saeed Seddighin.
\newblock Dynamic algorithms for {LIS} and distance to monotonicity.
\newblock In Konstantin Makarychev, Yury Makarychev, Madhur Tulsiani, Gautam
  Kamath, and Julia Chuzhoy, editors, {\em Proccedings of the 52nd Annual {ACM}
  {SIGACT} Symposium on Theory of Computing, {STOC} 2020, Chicago, IL, USA,
  June 22-26, 2020}, pages 671--684. {ACM}, 2020.

\bibitem[NR03]{NR03}
Fran\c{c}ois Nicolas and Eric Rivals.
\newblock Complexities of the centre and median string problems.
\newblock In {\em Proceedings of the 14th Annual Conference on Combinatorial
  Pattern Matching}, CPM'03, page 315–327, Berlin, Heidelberg, 2003.
  Springer-Verlag.

\bibitem[NR05]{nicolas2005hardness}
Fran{\c{c}}ois Nicolas and Eric Rivals.
\newblock Hardness results for the center and median string problems under the
  weighted and unweighted edit distances.
\newblock {\em Journal of discrete algorithms}, 3(2-4):390--415, 2005.

\bibitem[NSS17]{naumovitz2017accurate}
Timothy Naumovitz, Michael Saks, and C~Seshadhri.
\newblock Accurate and nearly optimal sublinear approximations to ulam
  distance.
\newblock In {\em Proceedings of the Twenty-Eighth Annual ACM-SIAM Symposium on
  Discrete Algorithms}, pages 2012--2031. SIAM, 2017.

\bibitem[Rub18]{R18}
Aviad Rubinstein.
\newblock Hardness of approximate nearest neighbor search.
\newblock In Ilias Diakonikolas, David Kempe, and Monika Henzinger, editors,
  {\em Proceedings of the 50th Annual {ACM} {SIGACT} Symposium on Theory of
  Computing, {STOC} 2018, Los Angeles, CA, USA, June 25-29, 2018}, pages
  1260--1268. {ACM}, 2018.

\bibitem[SAK{\etalchar{+}}01]{SAKSD01}
Kenneth~W. Shum, Ilia Aleshnikov, P.~Vijay Kumar, Henning Stichtenoth, and
  Vinay Deolalikar.
\newblock A low-complexity algorithm for the construction of
  algebraic-geometric codes better than the {Gilbert-Varshamov} bound.
\newblock {\em {IEEE} Trans. Information Theory}, 47(6):2225--2241, 2001.

\bibitem[SBG{\etalchar{+}}97]{SBGHM97}
Nikola Stojanovic, Piotr Berman, Deborah Gumucio, Ross~C. Hardison, and Webb
  Miller.
\newblock A linear-time algorithm for the 1-mismatch problem.
\newblock In Frank K. H.~A. Dehne, Andrew Rau{-}Chaplin,
  J{\"{o}}rg{-}R{\"{u}}diger Sack, and Roberto Tamassia, editors, {\em
  Algorithms and Data Structures, 5th International Workshop, {WADS} '97,
  Halifax, Nova Scotia, Canada, August 6-8, 1997, Proceedings}, volume 1272 of
  {\em Lecture Notes in Computer Science}, pages 126--135. Springer, 1997.

\bibitem[SP03]{SP03}
Jeong~Seop Sim and Kunsoo Park.
\newblock The consensus string problem for a metric is np-complete.
\newblock {\em J. Discrete Algorithms}, 1(1):111--117, 2003.

\bibitem[Wil18a]{W18}
Ryan Williams.
\newblock On the difference between closest, furthest, and orthogonal pairs:
  Nearly-linear vs barely-subquadratic complexity.
\newblock In Artur Czumaj, editor, {\em Proceedings of the Twenty-Ninth Annual
  {ACM-SIAM} Symposium on Discrete Algorithms, {SODA} 2018, New Orleans, LA,
  USA, January 7-10, 2018}, pages 1207--1215. {SIAM}, 2018.

\bibitem[Wil18b]{williams2018some}
Virginia~Vassilevska Williams.
\newblock On some fine-grained questions in algorithms and complexity.
\newblock In {\em Proceedings of the International Congress of Mathematicians:
  Rio de Janeiro 2018}, pages 3447--3487. World Scientific, 2018.

\end{thebibliography}
\appendix

\end{document}